%% file: bare_jrnl_new_sample4.tex
\documentclass[lettersize,journal]{IEEEtran}
\usepackage{amsmath,amsfonts}
\usepackage{algorithm}
\usepackage{array}
\usepackage[caption=false,font=normalsize,labelfont=sf,textfont=sf]{subfig}
\usepackage{textcomp}
\usepackage{stfloats}
\usepackage{url}
\usepackage{verbatim}
\usepackage{graphicx}
\usepackage{cite}
\usepackage{comment}
\usepackage{algpseudocode}
\algnewcommand{\LineIf}[2]{\State \algorithmicif\ #1 \algorithmicthen\ #2}
\usepackage{amsmath}
\usepackage{amssymb} 
\usepackage{amsthm}
\theoremstyle{definition} 
\newtheorem{definition}{Definition}[section] 
\newtheorem{lemma}{Lemma}[section] 
\newtheorem{theorem}{Theorem}[section] 
\usepackage{fancyhdr}

\begin{document}

\title{Catching the Fly: Practical Challenges \\in Making Blockchain FlyClient Real}

\author{Pericle Perazzo
\and
Dario Capecchi
\thanks{Pericle Perazzo and Dario Capecchi are with the Department of Information Engineering, University of Pisa, 56122 Pisa, Italy (e-mails: pericle.perazzo@unipi.it, dario.capecchi@ing.unipi.it).}%
}



\maketitle

\begin{center}
\textit{Submitted to IEEE Transactions on Dependable and Secure Computing}
\end{center}

\begin{abstract}
\input{sections/00_abstract}
\end{abstract}


\input{sections/01_introduction}
\input{sections/02_related_work}
\input{sections/03_background}
\input{sections/04_adversary_model}
\input{sections/05_prover}
\input{sections/06_optimizations}
\input{sections/07_conclusions}

\bibliographystyle{IEEEtran}
\bibliography{bibliography}

\end{document}

%% file: sections/00_abstract.tex
FlyClient is a lightweight blockchain verification protocol that enables proof-of-work validation using minimal data, making it ideal for resource-constrained environments like mobile wallets, Internet-of-Things devices or cross-chain bridges implemented with smart contracts. 
Despite its strong potential for enabling lightweight blockchain verification, FlyClient protocol is still in the experimental stages, with limited real-world deployments and performance evaluations under diverse conditions. 
In this paper we bridge the gap between theory and deployment, by addressing several technical challenges to advance FlyClient to a production-ready solution.
Namely, our contribution is three-fold: (i) we formally introduce an adversary model alternative to the original FlyClient one, that allows us to parametrize a verifier under a concrete economic interpretation, while also saving some proof space; 
(ii) we provide the first practical FlyClient prover implementation for a production blockchain (Zcash), and we estimate its performance under different configurations;
(iii) we introduce and evaluate two optimizations that minimize the size of FlyClient proofs, the first of which does not require any consensus change.

%% file: sections/01_introduction.tex
\section{Introduction}
\label{sec:introduction}

\IEEEPARstart{F}{lyClient} \cite{bunz2020flyclient,nemoz2021deployment,perazzo2024smartfly} is an innovative protocol designed to enable ultra-light clients for Proof-of-Work (PoW) blockchains. 
Traditional blockchain clients require downloading and verifying large amounts of data to ensure security and correctness, which is impractical for resource-constrained execution environments. 
FlyClient addresses this challenge 
by leveraging probabilistic sampling of block headers combined with Merkle Mountain Range (MMR) commitments, thus producing proofs of chain validity that grow only logarithmically with the chain length.
The FlyClient protocol allows clients to verify the entire blockchain with minimal bandwidth and storage requirements, while maintaining strong security guarantees. 
It can produce both interactive or non-interactive proofs, and it can also deal with PoW difficulty variability. 
FlyClient finds application in those cases in which the state of the blockchain must be verified by execution environments with limited resources. 
A straightforward application is mobile wallets, which must make sure the presence of given transactions inside the blockchain. 
Current mobile wallets trust third-party full nodes to maintain and check the whole blockchain, but FlyClient can do the same without any trust assumption. 
Another application is in the Internet of Things (IoT), in which resource-constrained devices can use FlyClient to perform machine-to-machine transactions or read and write immutable logs. 
A third application is in blockchain bridges, which are aimed at moving a crypto-asset from a blockchain to another blockchain and back. 
Current blockchain bridges employ trusted third-party nodes or federations of nodes, but FlyClient allows a smart contract to verify the state of a foreign blockchain by validating a non-interactive FlyClient proof. 
This permits trustless bridges, in which the correctness of a cross-chain interaction is automatically proved without trusting any (particular) node. 
A blockchain can support the FlyClient functionalities natively \cite{nemoz2021deployment} or by means of smart contracts \cite{perazzo2024smartfly}.
Among the cryptocurrencies that natively support FlyClient, Zcash \cite{zcash} is the most important and capitalized one. 
Zcash integrated FlyClient through ZIP 221 \cite{zip221} as part of the Heartwood network upgrade.
Despite this, no FlyClient prover for Zcash has been developed yet, and the FlyClient protocol is still in the experimental stages. 

This paper addresses several technical challenges to advance FlyClient from a theoretical concept to a production-ready solution. 
In particular, its contribution is three-fold. 
First, we introduce the \emph{$w_a$-adversary model} as an alternative to the original FlyClient's $(c,L)$-adversary model. 
The $w_a$-adversary is based on a single parameter that directly captures the budget available to the attacker for performing his attack, and thus it has a concrete economic interpretation useful for real deployments. 
A FlyClient verifier can be parametrized by assuming a certain $w_a$-adversary and solving an optimization problem, which also offers some proof size reduction that we measure. 
Second, we present our experiences in extending Zebrad \cite{zebrad}, which is the Rust-based official full node implementation of Zcash, for offering FlyClient prover functionalities. 
We also estimate the extra resources that are needed for such functionalities, and we measure the proof size under different representations and different verifier variants.
The prover implementation is made available to the Zebrad community, which, at the time of writing, has already accepted some of our extensions in the Zebrad official repository.
Third, we introduce and evaluate two optimizations that aim at minimizing the size of the interactive and non-interactive FlyClient proofs.
Namely, we introduce the \emph{cumulative proofs} and the \emph{distilled proofs}. 
The former optimization offers limited advantage (-9.16\% non-interactive proof size), but it is already deployable at zero cost. 
The latter one requires a change in the consensus layer, but it offers a more substantial advantage (-71.33\% non-interactive proof size).

The rest of the paper is structured as follows. 
Section \ref{sec:rw} discusses related work. 
Section \ref{sec:background} explains the necessary background concepts. 
Section \ref{sec:adversary_model} introduces the $w_a$-adversary model and the relative configuration method for the verifier. 
Section \ref{sec:prover} introduces and evaluates our FlyClient prover implementation for Zcash. 
Section \ref{sec:optimizations} introduces and evaluates our proposed FlyClient proof optimizations. 
Finally, the paper is concluded in Section \ref{sec:conclusions}.

%% file: sections/02_related_work.tex
\section{Related Work}
\label{sec:rw}

\subsection{Research on FlyClient}
B\"unz et al. \cite{bunz2020flyclient} introduced the FlyClient protocol to enable extremely lightweight verification of PoW blockchains by reducing the data requirements for clients from linear to logarithmic complexity. 
FlyClient achieves efficiency through probabilistic sampling and MMR commitments embedded in block headers, allowing clients to verify the entire chain with only a logarithmic number of sampled blocks, even in the case of variable-difficulty blockchains. 
The authors model the adversary as a $(c,L)$-adversary, that is an attacker that cannot produce an $L$-long fork with $c$ validity ratio. 
Although this modeling is convenient to mathematically prove the security properties of FlyClient, it lacks a direct economic interpretation, thus it is less convenient from the point of view of a real deployment.
The present paper 
introduces the alternative $w_a$-adversary model, which is based on a single parameter that directly captures the budget available to the attacker for performing his attack, and thus it has a concrete economic interpretation. 

Nemoz and Zamyatin \cite{nemoz2021deployment} explore the security implications of deploying FlyClient using a velvet fork, which allows backward compatibility without requiring all nodes to upgrade. 
In particular, the authors explored the possibility of chain-sewing attacks (originally introduced by \cite{kiayias2021velvet}), in which an adversary stitches together parts of different chains to deceive FlyClient verifiers about the state of the blockchain. 
This attack is feasible because, in a velvet fork, a block could contain invalid FlyClient data without being rejected by the nodes. 
However, this does not apply to Zcash, which implemented FlyClient via a hard fork. 
In the Zcash case, all blocks are required to include valid FlyClient data, hence preventing the conditions necessary for the chain-sewing attack to succeed.
Our FlyClient prover implementation is therefore immune to chain-sewing attacks.

Perazzo and Xefraj \cite{perazzo2024smartfly} proposed SmartFly, a FlyClient implementation via smart contracts which does not require forks at all. 
The nodes that invoke the smart contract methods do not need to be trusted, as long as the smart contract language can access the hash of the most recent block. 
SmartFly does not need changes in the blockchain format or forks of any kind, but it requires some monetary cost to be maintained, due to the invocations of smart contract methods. 
Our FlyClient prover is based on Zcash, which implemented FlyClient via a hard fork, so it does not require any monetary cost to be maintained.

\subsection{FlyClient Alternatives}
Kiayias et al. \cite{kiayias2016proofs} and Kiayias et al. \cite{kiayias2017noninteractive} proposed PoPoW and NIPoPoW, which first enabled sublinear (logarithmic) verification for light clients of PoW blockchains. 
These techniques are based on superblocks, that are blocks with a PoW difficulty higher than the difficulty target. 
Unfortunately, PoPoW and NIPoPoW assume a constant difficulty, and they do not apply well in real blockchains where difficulty is variable and automatically adjusted from the mining timestamps. 
FlyClient is able to produce logarithmic-size proofs as well, but at the same time it can also deal with PoW difficulty variations.

Bonneau et al. \cite{bonneau2020coda} proposed Coda, which uses recursive SNARK proofs to compress the whole blockchain into a single constant-size proof, so that its validity can be verified in constant time. 
Despite being promising from the point of view of the verifier's performance, Coda requires a high computational cost for block producers, and a complete re-design of the blockchain protocol, hence it is not easily applicable to already existing cryptocurrencies.
Vesely et al. \cite{vesely2022plumo} use (non-recursive) SNARK proofs to implement ultralight clients. 
With respect to Coda, Plumo is applicable to already existing blockchains. 
Moreover, it is consensus-agnostic, in the sense that it can be applied independently of the kind of employed consensus: PoW or PoS or BFT. 
Plumo enjoys thus a larger applicability with respect to FlyClient, which applies only to PoW blockchains. 
However, Plumo incurs significantly higher prover-side computational costs due to proof generation, compared to FlyClient which requires only lightweight operations such as MMR maintenance and header serving. 
Interestingly, SNARKs and FlyClient are not mutually exclusive solutions, in the sense that a non-interactive FlyClient proof could in principle be ``compressed'' with SNARK techniques. 
We plan to explore such a possibility in future work.

%% file: sections/03_background.tex
\section{Background}
\label{sec:background}

\subsection{MMR and Proofs of Ancestry} 
The FlyClient protocol \cite{bunz2020flyclient} is based on the Merkle Mountain Range (MMR) data structure \cite{mmr}, which is an extension of a Merkle tree that allows for efficient (logarithmic) appends. 
Fig. \ref{fig:mmr} exemplifies an MMR constructed over a blockchain. 
\begin{figure}
    \centering
    \includegraphics[trim={1,12cm 7.25cm 19.19cm 2.36cm},clip,width=0.9\linewidth]{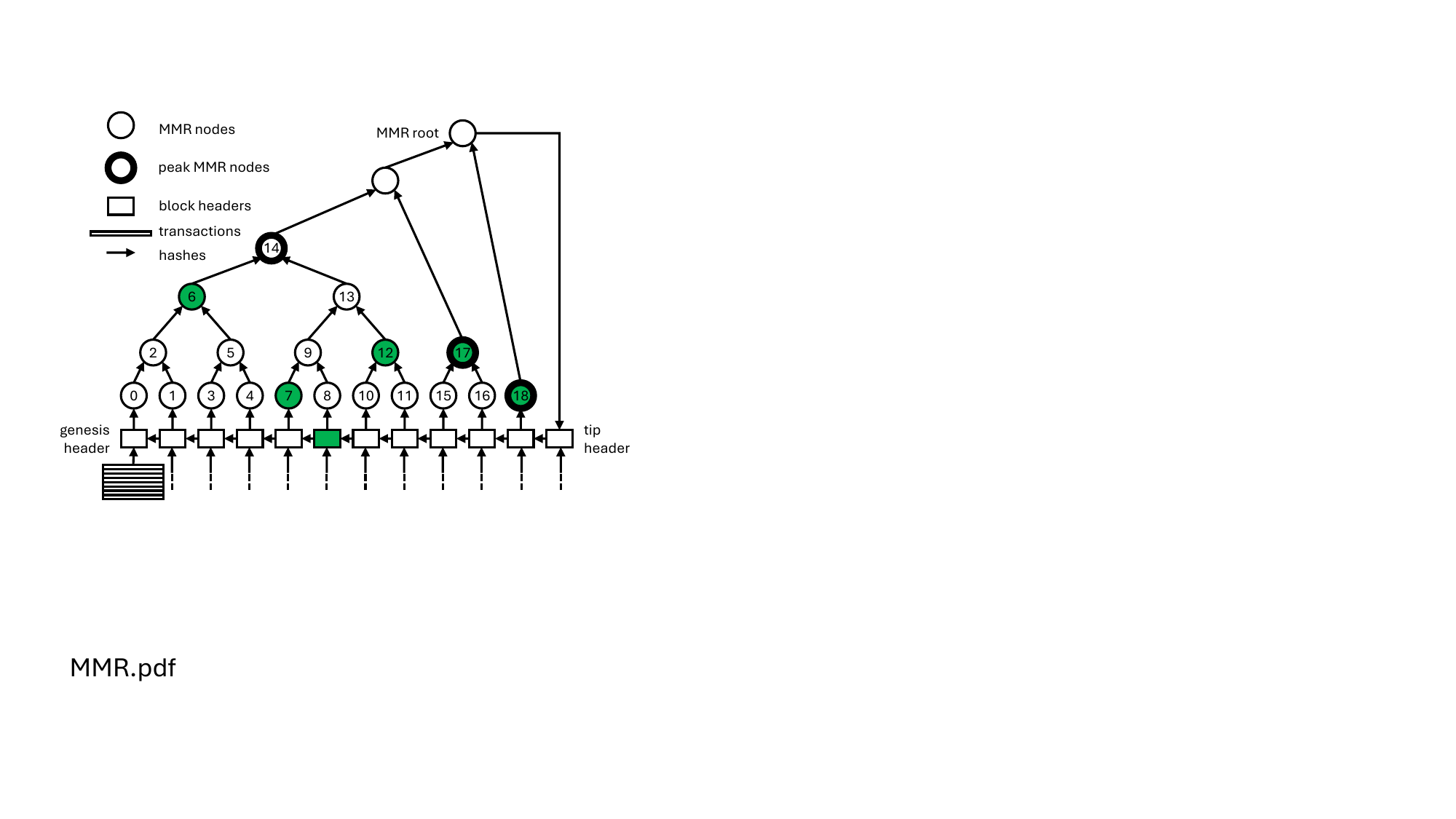}
    \caption{Example of MMR constructed over a blockchain.}
    \label{fig:mmr}
\end{figure}
The blockchain is an ordered sequence of \emph{blocks}, each of which is composed by a small \emph{block header} (or simply \emph{header}) and a potentially big number of ordered transactions. 
Each header contains a commitment of all the block's transactions, which in the figure is simplified by a single hash, but in real implementations is a Merkle tree root to allow for efficient inclusion proofs. 
Each header contains the hash of the previous header (\emph{previous hash}), thus composing the so-called \emph{header chain}. 
The first header is called \emph{genesis header}. 
It is well-known by all the protocol participants and it contains a null previous hash. 
The (currently) last header is called \emph{tip header}. 
The MMR is built starting from the leaf nodes, each of which is relative to a header and contains its hash. 
Each intermediate node of the MMR contains the hash of both its children. 
In general, the MMR is not balanced, but it can be seen as a sequence of balanced Merkle trees (the mountains), whose root nodes are called \emph{peak MMR nodes}, or simply \emph{peaks}. 
The peaks and their descendants are persistent, in the sense that they make part of the MMR forever, independently on how many appends are done. 
They have a unique numerical identifier, which reflects the order of their creation. 
The MMR nodes over the peaks are called \emph{generated MMR nodes}, they are not persistent, and they do not have any identifier. 
The purpose of the generated MMR nodes is to reduce the peaks into a single root node, whose hash is then stored in the header successive to the last one covered by the MMR. 
Each MMR node contains all the pieces of information needed to check the plausibility of the difficulty transitions between the blocks covered by the node itself. 
Typically, it contains the timestamps and the difficulties of the first and the last covered block, and the cumulative work of all the covered blocks.

It is possible to build a logarithmic-size \emph{proof of ancestry} that proves that a given MMR is an ancestor of another MMR, that is the latter has been obtained from the former with a given number of appends. 
Each MMR is represented by the header that contains its root hash. 
The proof of ancestry of header $i$ with respect to header $j>i$ is constituted by the minimal set of nodes of the MMR of header $j$ that cover exactly all the headers except header $i$. 
For example, the proof of ancestry of header $i=6$ with respect to the tip header $j=11$ in Fig. \ref{fig:mmr} is constituted by the green MMR nodes.

\subsection{FlyClient Protocol}
FlyClient \cite{bunz2020flyclient} is a super-light client protocol designed for Proof-of-Work (PoW) blockchains. 
Its goal is to allow resource-constrained devices like mobile wallets or IoT devices to verify the validity of an entire blockchain without downloading all block headers or trusting any third-party full node. 
By using the FlyClient protocol, a prover can convince a verifier to own a header chain of a given length, and that these headers are valid, meaning that they all contain a valid proof of work and the difficulty transitions between them are correct. 
To do this without downloading all the headers, the verifier first downloads the tip header from the prover, then it \emph{samples} a logarithmic number of headers. 
``Sampling a header'' means (i) downloading it, (ii) checking its proof of work to be valid, (iii) downloading and checking its ancestry proof with respect to the tip header, (iv) checking the plausibility of the difficulty transitions via the MMR nodes of the ancestry proof. 
The verifier deterministically samples the last $n_\mathit{det}$ headers precedent to the tip header, and $n_\mathit{prob}$ headers probabilistically chosen among the precedent ones, using a specific probability function biased towards the rightmost headers.

To parametrize $n_\mathit{det}$ and $n_\mathit{prob}$ we have to assume a $(c,L)$-adversary, that is a malicious prover that cannot produce a fork longer than $L$ with $c$ validity ratio, except with probability negligible with the security parameter $\lambda$.
If a verifier wants to guarantee that such an adversary has less than $2^{-\lambda}$ probability to succeed in his attack, then it can parametrize the FlyClient protocol with:
\begin{eqnarray}
\label{eqn:interactive_n_det}
n_\mathit{det} &=& L \\
\label{eqn:interactive_n_prob}
n_\mathit{prob} &=& \frac{\lambda}{\log_{0.5} \left(1-\frac{1}{\log_c \delta}\right)},
\end{eqnarray}
where $n$ is the chain length, and $\delta$ is the fraction of work sampled deterministically over the total chain work. 
In general, $\delta = \frac{\sum_{i = n - L + 1}^n d_i}{\sum_{i = 1}^n d_i}$ where $d_i$ is the difficulty of the $i$-th block of the chain, while in case of fixed difficulty it simplifies in $\delta = L / n$.
We call \emph{FlyClient proof} the set of headers and MMR nodes downloaded by the verifier during the FlyClient protocol. 
The FlyClient proof convinces the verifier that the prover owns a valid header chain of a given length.

The FlyClient protocol has a non-interactive version, which allows the prover to be offline while the verifier checks the FlyClient proof. 
In this way the FlyClient proof can be conveniently stored on a file, a database, or a smart contract. 
To produce a non-interactive FlyClient proof, a prover must apply the Fiat-Shamir heuristic\cite{fiat1986how} to the tip header, and increase the protocol parameters in order to compensate the relative decrease in protocol security.
With the same $(c,L)$-adversary and the same $2^{-\lambda}$ target probability of adversarial success, the verifier must parametrize the non-interactive FlyClient proof with:
\begin{eqnarray}
\label{eqn:noninteractive_n_det}
n_\mathit{det} &=& L \\
\label{eqn:noninteractive_n_prob}
n_\mathit{prob} &=& \frac{\lambda - \log_{0.5} (c \cdot n)}{\log_{0.5} \left(1-\frac{1}{\log_c \delta}\right)}.
\end{eqnarray}

\subsection{Zcash Implementation of FlyClient} 
Zcash \cite{zcash} implemented FlyClient through ZIP 221 \cite{zip221} as part of the Heartwood network upgrade, with a hard fork happened on July 2020. 
Since then, as of April 2026 Zcash underwent four more consensus changes called \emph{network upgrades}, each implemented with a hard fork. 
Namely the network upgrades have been Canopy (November 2020), NU5 (2022), NU6 (2024), and the minor upgrade NU6.1 (November 2025). 
Some of them changed the format and/or the meaning of the fields in the headers and the MMR nodes. 
The part of the blockchain that uses the same consensus rules between two network upgrades is called \emph{consensus branch}. 
Notably, the MMR does not grow indefinitely in Zcash, but rather its construction reboots from scratch at each network upgrade. 
In practice each consensus branch has its own MMR independent to the precedent ones, and the first header of each consensus branch constitutes the connection point between the old and the new MMR.
This multi-MMR approach deviates from the original FlyClient mechanism \cite{bunz2020flyclient} and it guarantees the cryptographic independency of different consensus branches.

In this section we refer to the Zcash protocol version of NU6, as the newer NU6.1 did not change anything in the FlyClient-related mechanisms.
Each node of the MMR maintained by the Zcash protocol has the following fields:
\begin{itemize}
\item \texttt{hashSubtreeCommitment}: the hash of the corresponding block header in case the node is a leaf, or the hash of the right and left child nodes serialized and concatenated together in case the node is not a leaf. It is represented on 32 bytes.
\item \texttt{nEarliestTimestamp}, \texttt{nLatestTimestamp}: the mining timestamps respectively of the first and the last block covered by the node, which are the same block in case the node is a leaf. They are represented each on 4 bytes.
\item \texttt{nEarliestTargetBits}, \texttt{nLatestTargetBits}: the difficulty targets respectively of the first and the last block covered by the node. They are represented each on 4 bytes.
\item \texttt{nEarliestHeight}, \texttt{nLatestHeight}: the heights respectively of the first and the last block covered by the node. They are represented each with a variable-length format, whose size varies from 1 (for small numbers) to 9 bytes (for big numbers).
\item \texttt{nSubTreeTotalWork}: the cumulative difficulty of the blocks covered by the node. It is represented on 32 bytes.
\end{itemize}
There are some additional fields that are not used directly by the FlyClient protocol, but rather they serve different security or efficiency purposes that are foreseeable for super-light clients. 
The total size of a Zcash MMR node varies between 212 and 244 bytes.

Each block header contains the following fields:
\begin{itemize}
  \item \texttt{hashBlockCommitments}: a cryptographic commitment of \texttt{hashAuthDataRoot} and \texttt{hashChainHistoryRoot}, which in turn are a commitment to all transaction authorization data in the block, and a commitment to the root of the MMR built from the genesis to the block previous to this. It is represented on 32 bytes.
  \item \texttt{nTime}: the block timestamp encoded as a Unix epoch time. It is represented on 4 bytes.
  \item \texttt{nBits}: the target threshold that the block header hash must satisfy for proof of work. It is represented on 4 bytes.
  \item \texttt{nNonce}: an arbitrary field that miners vary in order to search for a header hash meeting the target threshold. It is represented on 32 bytes.
  \item \texttt{solutionSize}: an integer on 3 bytes, whose value is always 1344.
  \item \texttt{solution}: the Equihash solution included in the header and required to be valid under the consensus rules. It is represented on 1344 bytes.
\end{itemize}

%% file: sections/04_adversary_model.tex
\section{Adversary Model and FlyClient Parametrization}
\label{sec:adversary_model}

The original FlyClient paper \cite{bunz2020flyclient} models the adversary as a $(c,L)$-adversary, that is an attacker that cannot produce a fork longer than $L$ with $c$ validity ratio, except with negligible probability. 
This modeling style is an extension of the classic adversary model in blockchain research, which is based on the share of resources (hash power or stake) controlled by the adversary. 
However, it is not convenient from the point of view of concrete deployment, because it is unclear how a realistic attacker can fit the two parameters $L$ and $c$, and this eventually leads to an arbitrary parametrization of the FlyClient verifier. 
In this paper we introduce the $w_a$-adversary model, which is easier to understand because it is based on a single parameter that directly captures the expected budget available to the attacker for performing his attack, and thus it has a straightforward and easy interpretation. 
We show that the $w_a$-adversary model enables direct mapping from expected attacker budget to verifier parameters, achieving the same security of ($c,L$)-based configurations while simplifying deployment and also reducing the proof size as side effect.

\subsection{$w_a$-Adversary}
\label{sec:wa_adversary}
Assuming a $w_a$-adversary means assuming that the attacker cannot produce blocks with more than $w_a$ cumulative work with non-negligible probability, while respecting the difficulty transition rules. 
We call $w_a$ the \emph{adversarial work budget}. 
More formally, we state the following. \begin{definition}[$w_a$-adversary]
An adversary $\mathcal{A}$ is a $w_a$-adversary if, in any attack execution against a FlyClient verifier, the valid malicious blocks that $\mathcal{A}$ can generate while respecting the consensus rules have at most $w_a$ cumulative work, except with probability negligible with a security parameter.
\end{definition}
The $w_a$-adversary has a direct monetary interpretation, because we can compute the expected cost of producing a valid block in terms of machinery and energy consumption. 
In general, given the expected cost $C_{51\%/h}$ of mounting a 51\% attack (that is the cost of producing the same hashrate of the honest network) for 1 hour, and given the expected number of blocks $n_\mathit{honest/h}$ produced by the honest network in 1 hour, the current average block difficulty $\tilde{d}$, and the adversarial work budget $w_a$, we can compute the expected adversarial budget as $\mathit{Budget}_a = C_{51\%/h} \cdot \frac{w_a}{\tilde{d} \cdot n_\mathit{honest/h}}$.
For example, at the time of writing performing a 51\% attack against Zcash for one hour is estimated to cost about 20,000 USD\footnote{\url{https://www.crypto51.app/} (accessed on 5 April 2026).}. 
Considering that the Zcash network produces 48 blocks each hour in average, and assuming an adversarial work budget of 45,000 GH\footnote{In this paper we conveniently measure the work in GH (GigaHash), that is the expected number of hashes that have been computed by the miner to produce the block(s).} and a current average block difficulty of 900 GH, the expected adversarial budget to produce them is thus 20,833 USD.
Of course a similar estimation could be done for other PoW-based blockchains as well.
Note however that this estimate could be optimistic, 
for two reasons. 
First of all because the computed budget is only an expected one, not a real one. 
Indeed, the attacker could be lucky in finding valid blocks, and thus he could produce more than 50 blocks with the same budget. 
Secondly, producing older valid blocks could be cheaper than producing newer ones, due to the tendential rise of the PoW difficulty. 
To mitigate this second source of optimism, one should consider the cost of producing an average block of the blockchain, from the genesis to the tip.

\subsection{FlyClient Parametrization}
We then show how to parameterize a FlyClient verifier based on such an adversary model. 
According to the definition, a $w_a$-adversary cannot create a fork of valid blocks with more than $w_a$ work in them. 
Therefore, such an adversary fulfills a traditional $(c,L)$-adversary assumption with $c=1$ and $L = \min \left\lbrace \nu \middle| \sum_{i = n - \nu + 1}^{n} d_i \ge w_a \right\rbrace$. 
In addition, he cannot create a fork of blocks with more than $2 w_a$ work in them, half of which is valid, so he fulfills also another $(c,L)$-adversary assumption this time with $c=0.5$ and $L = \min \left\lbrace \nu \middle| \sum_{i = n - \nu + 1}^{n} d_i \ge 2 w_a \right\rbrace$. 
In general, a $w_a$-adversary assumption automatically fulfills a \emph{family} of $(c,L)$-adversary assumptions having the property:
\begin{equation}
\label{eq:c_L_family}
L \ge \min \left\lbrace \nu \middle| \sum_{i = n - \nu + 1}^{n} d_i \ge w_a / c \right\rbrace.
\end{equation}
More formally, we state the following:
\begin{lemma}[Reduction to a family of $(c,L)$-adversaries]
\label{thm:c_L_family}
Let $\mathcal{A}$ be a $w_a$-adversary. 
Then, for every possible $c \in [0,1]$ and $L>0$ such that Eq. \ref{eq:c_L_family} is respected, such an adversary cannot violate the $(c,L)$-adversary assumption.
\end{lemma}
\begin{proof}[Proof sketch]
By definition, a $w_a$-adversary can produce malicious blocks with at most $w_a$ valid work in them, except with negligible probability. 
Consider any possible fork produced by $\mathcal{A}$ with length $L_\mathcal{A}$ and validity ratio $c_\mathcal{A}$. 
Only two cases are possible. 
In case $c_\mathcal{A} < c$ then the $(c,L)$-adversary assumption is not violated, since the fork has an insufficient validity ratio. 
On the other hand, in case $c_\mathcal{A} \ge c$, $\mathcal{A}$ produced a fork with a correct validity ratio, but with an insufficient length $L_\mathcal{A} \le \min \left\lbrace \nu \middle| \sum_{i = n - \nu + 1}^{n} d_i \ge w_a / c_\mathcal{A} \right\rbrace \le \min \left\lbrace \nu \middle| \sum_{i = n - \nu + 1}^{n} d_i \ge w_a / c \right\rbrace \le L$. 
Thus even in this case the $(c,L)$-adversary assumption is not violated.
\end{proof}
Let us suppose a verifier that is configured according to \emph{any} of these $(c,L)$-adversary assumptions. 
Successfully attacking such a verifier requires to violate the $(c,L)$-adversary assumption, which in turn requires to violate the $w_a$-adversary assumption. 
This in practice makes the $(c,L)$ pair a degree of freedom that we can use to configure conveniently our verifier.
In other words, to be secure against the $w_a$-adversary, the verifier can be configured to be secure against \emph{any} of the possible $(c,L)$-adversaries of the family specified by Eq. \ref{eq:c_L_family}, chosen at will. 
More formally, we state the following.
\begin{theorem}[Security of the $w_a$-parametrized verifier]
Assuming a chain length $n$ and a FlyClient verifier $\mathcal{V}$ parametrized according to Eqq. 1 and 2 (Eqq. 3 and 4 in case of non-interactive) with $(c,L)$ such that Eq. \ref{eq:c_L_family} is respected, and with security parameter $\lambda$, then a $w_a$-adversary cannot make $\mathcal{V}$ accept an invalid chain with probability more than negligible with the security parameter.
\end{theorem}
\begin{proof}[Proof sketch]
By the Lemma \ref{thm:c_L_family}, the $w_a$-adversary cannot violate the $(c,L)$-adversary assumption. 
Thus, Theorem 1 in \cite{bunz2020flyclient} (security of FlyClient) directly yields the claim.
\end{proof}
Normally, the best thing would be to choose $c$ and $L$ in such a way to minimize the proof size, but this would be overly complex and would depend on the sizes and the representation format of the block headers and the MMR nodes, which are specific for each blockchain. 
We follow a more general approach, that is to minimize the total number of sampled headers. 
This approach abstracts away from the specific blockchain and approximates well the minimization of the proof size in those cases in which the size of the block headers dominates that of the MMR nodes, like in Zcash.

Following the B\"unz-Kiffer parametrization under a $(c,L)$-adversary assumption, the verifier must perform $n_\mathit{det}$ deterministic samplings plus $n_\mathit{prob}$ probabilistic ones, 
with $n_{det}$ and $n_{prob}$ chosen as in Eqq. \ref{eqn:interactive_n_det} and \ref{eqn:interactive_n_prob} or Eqq. \ref{eqn:noninteractive_n_det} and \ref{eqn:noninteractive_n_prob}, respectively if we are running the interactive or the non-interactive version of the protocol.
Therefore, we find the $(c = c^*,L = L^*)$ pair that minimizes $n_\mathit{det} + n_\mathit{prob}$, with $L^* \ge \min \left\lbrace \nu \middle| \sum_{i = n - \nu + 1}^{n} d_i \ge w_a / c^* \right\rbrace$. 
Crucially, in case the block difficulty has low variability like it happens in Zcash and other blockchains, we can assume that the average block difficulty of the last $L^*$ blocks ($\tilde{d}$) does not vary with $L^*$. 
In formulas: $\forall L^*: 1/L^* \cdot \sum_{i = n - L^* + 1}^{n} d_i = \tilde{d}$. 
Therefore, the relation between $c^*$ and $L^*$ simplifies in $c^* \cdot L^* = w_a / \tilde{d} \triangleq n_a$. 
We call $n_a$ the \emph{adversarial block budget}, which represents the budget of the adversary measured in terms of number of blocks instead of amount of work.
By formally computing the derivative of $n_\mathit{det} + n_\mathit{prob}$ over $c^*$ and imposing it to zero, we find that $c^*$ must fulfill the following first-order conditions:
\begin{equation}
\label{eqn:interactive_opt_c}
\begin{aligned}
\lambda (c^*)^2 \ln 2
\left(
-\frac{\ln c^*}{a^2 c^*}
-\frac{1}{a c^*}
\right)
&= n_a
\left(
1-\frac{\ln c^*}{a}
\right) \\
&\quad \cdot
\ln^2
\left(
1-\frac{\ln c^*}{a}
\right),
\end{aligned}
\end{equation}
for the interactive case, and:
\begin{equation}
\label{eqn:noninteractive_opt_eq}
\left(\ln \frac{b}{a} \cdot n_a + c^*\right) \ln \frac{b}{a} \cdot a \cdot b + c^* \ln \frac{n_a}{n} \cdot \left(\lambda \ln 2 +\ln(c^* \cdot n)\right) = 0
\end{equation}
for the non-interactive one, where $a = \ln \left(\frac{n_a}{c^* n}\right)$, $b = \ln \left(\frac{n_a}{(c^*)^2 n} \right)$, and $c^* \in [0,1]$. 
The solutions to such equations are not expressible in closed forms, so they must be found numerically. 
Figure \ref{fig:opt_c} shows the trend of $c^*$ with a security level of $\lambda=50$, and varying the blockchain length and the adversarial block budget, for both the interactive and non-interactive protocols. 
\begin{figure}[t]
    \centering
    \includegraphics[width=0.9\linewidth]{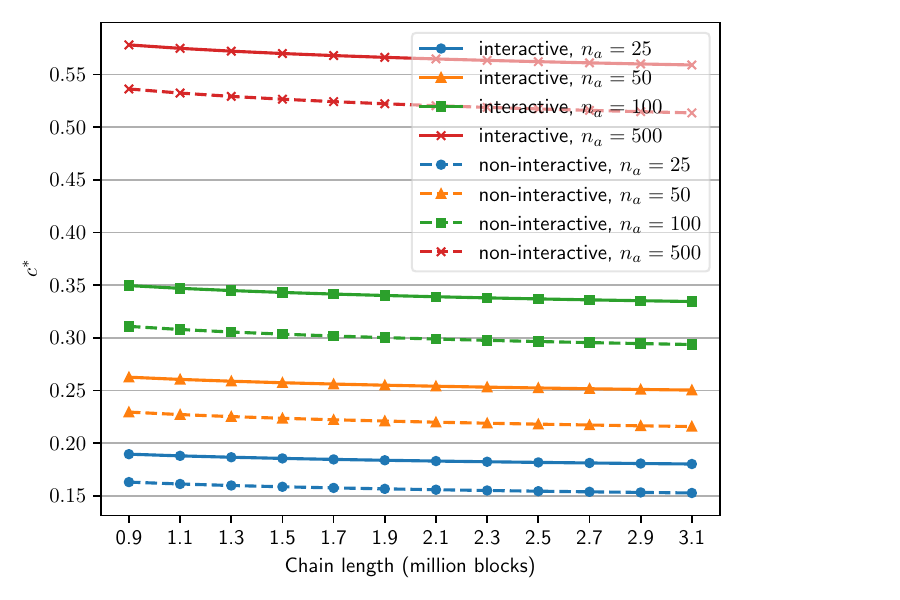}
    \caption{Optimal $c^*$ for interactive and non-interactive proofs, varying the blockchain length and the adversarial block budget.}
    \label{fig:opt_c}
\end{figure}
The optimal $L^*$ can then be computed simply as $L^* = n_a/c^*$.

As a concrete example, let us consider a FlyClient verifier configured with a traditional $(c,L)$-adversary model with $c=0.5$ and $L=100$, a target security level $\lambda=50$, and a chain of 3 million blocks. 
Breaking such an assumption requires an expected adversarial block budget of $n_a = c \cdot L = 50$ blocks, which in turn requires an expected budget of $\mathit{Budget}_a = 20,833$ USD, as computed in Section \ref{sec:wa_adversary}.
This configuration results in an interactive FlyClient proof of 598 total samplings, and a non-interactive one of 802 total samplings. 
On the other hand, by using a $w_a$-adversary model with the same expected block budget ($n_a = c \cdot L = 50$), we obtain an optimal $(c^*, L^*)$ pair of (0.25, 200) for the interactive FlyClient proof, and of (0.2161, 231) for the non-interactive one. 
This configuration results in an interactive proof of 423 samplings and a non-interactive one of 504 samplings (-37\% saving). 
Considering that each Zcash header is 1,487 bytes if represented in binary format, the $w_a$-adversary parametrization decreases the amount of total header data to download from 598 $\times$ 1,487 bytes = 0.848 MiB to 423 $\times$ 1,487 bytes = 0.600 MiB for the interactive case (-29\% saving), and from 802 $\times$ 1,487 = 1.137 MiB to 504 $\times$ 1,487 = 0.715 MiB for the non-interactive one (-37\% saving). 
Note that this computation does not consider the MMR nodes that the verifier must download together with the block headers, but these grow roughly linearly with the number of samplings, so the percentages of saving should be confirmed. 
In Section \ref{sec:prover} we will consider the MMR nodes too in the FlyClient proof size. 
This change in the adversary model does not appear to introduce additional security risks under the same adversarial budget assumptions, because as explained above an attacker producing a $<100$-block long fork with 0.5 validity ratio does not have the necessary budget to violate the $(c^*=0.25,L^*=200)$-adversary or the $(c^*=0.2161, L^*=231)$-adversary assumptions, so the verifier also defends against such adversaries.

%% file: sections/05_prover.tex
\section{Prover Implementation}
\label{sec:prover}
In this section we estimate the extra resources that are needed for a Zcash full node for offering FlyClient prover functionalities. 
To do this, we extended Zebrad\footnote{\url{https://github.com/ZcashFoundation/zebra} (accessed on 25 March 2026).}, which is the Rust-based official full node implementation of Zcash, developed by the Zcash Foundation as part of the Zebra project. 
Zebrad replaces the legacy Zcashd full node implementation, written in C++ and based on Bitcoin Core, which became obsolete in 2025\footnote{\url{https://z.cash/support/zcashd-deprecation/} (accessed on 25 March 2026).}. 
Such a FlyClient-prover extension of Zebrad may have independent value for future research on FlyClient, so we made it available to the Zcash community\footnote{\url{https://github.com/Metalcape/zcash-flyclient} (accessed on 25 March 2026).}.

We conveniently chose to put all the protocol's intelligence on the verifier side, that is the prover is only a dumb answerer to the verifier's requests, performed via Remote Procedure Calls (RPCs). 
Zebrad implements an HTTP-based JSON-RPC protocol like Bitcoin, by which a client can perform block queries or send transactions to be included in the blockchain. 
The RPC functions that the prover must implement are the following. 
\begin{itemize}
\item A function to get the tip header, the number of blocks, and the total work from the genesis to the tip header. 
We used the native function \textproc{getBlockchainInfo} for this.
\item A function to get a block header given its height in the chain. 
This is used by the deterministic and the probabilistic samplings. 
We used the native function \textproc{getBlockHeader} for this.
\item A function to get an MMR node given its index. This is used to retrieve the MMR proofs of the sampled blocks. 
We introduced a new function \textproc{getHistoryNode} for this.
\item A function to get the \texttt{authDataRoot} of a block given its height. 
This is necessary to allow the verification of MMR proofs. We introduced a new function \textproc{getAuthDataRoot} for this.
\item A function to get the total work from the genesis to a given block height. 
This is used to retrieve the total work associated with the oldest deterministically sampled block, which in turn allows us to compute the probability distribution by which we perform the probabilistic sampling in case of variable difficulty. 
We introduced a new function \textproc{getTotalWork} for this.
\item A function to get the height of the first block having the total work equal to or greater than a given amount. 
This is used by probabilistic sampling in case of variable difficulty. 
We introduced a new function \textproc{getHeightWithTotalWork} for this.
\end{itemize}
The most resource-consuming RPC function is \textproc{getHistoryNode}, because it needs the prover to maintain a database of MMR nodes. 
The current Zebrad implementation lacks this functionality, and it maintains in memory only a minimal fraction of MMR nodes, namely the MMR peaks, which are needed to calculate the chain history root of the next block. 
We implemented the routines to populate the MMR nodes database contextually with the download of the blocks or after it. The first routine is performed by an entity which installs Zebrad from scratch. 
The second routine is performed by an entity which already has a running Zebrad node synchronized with the current blockchain, but without prover functionalities, and it wants to offer such functionalities too. 
We measured the time needed to upgrade the database to store the individual MMR tree nodes and compared it to the time needed to fully synchronize with the blockchain, with both a legacy and a FlyClient-capable Zebrad instance. 
We run the three tests on the same virtual machine with high-bandwidth network access. 

Figure \ref{fig:db_upgrade_time} shows the results. 
\begin{figure}[t]
    \centering
    \includegraphics[width=0.9\linewidth]{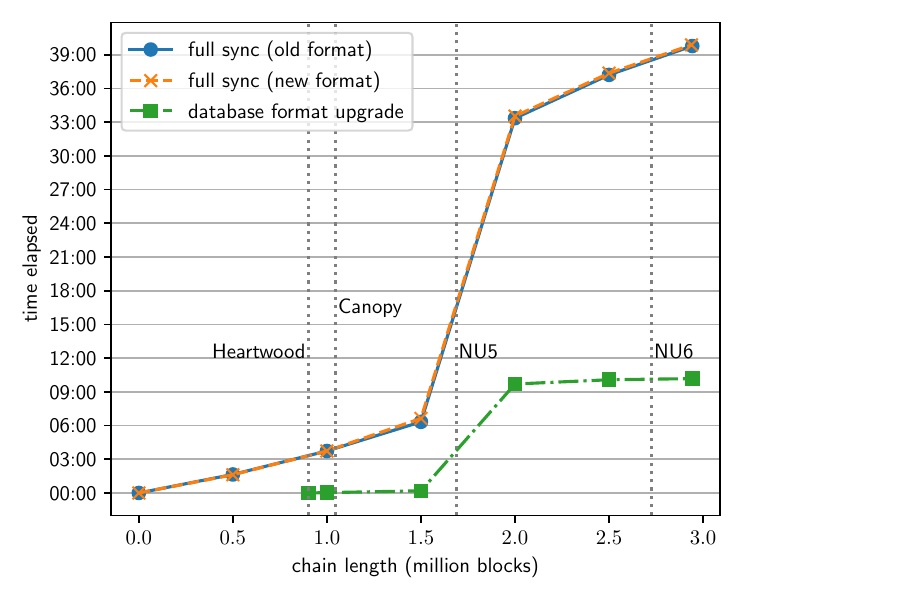}
    \caption{Database upgrade time.}
    \label{fig:db_upgrade_time}
\end{figure}
The time for a full sync of a FlyClient-capable Zebrad instance is comparable to a legacy one. 
This is because the impact of the generation and writing of MMR nodes to the database on the total execution time of the full sync is negligible. 
The only additional operation needed by the modified full node is the database write operation, because the nodes already had to be calculated to keep the HistoryTree of the tip header up to date. 
Therefore, FlyClient capability does not introduce sensitive degradation with respect to the sync time. 
If a legacy instance with an old database format gets upgraded, it takes about 10 hours to upgrade the database to the new format with a chain of about 3M blocks, while a full sync requires almost 40 hours.

It can be seen from the graph that the upgrade time spikes around 1.5 million blocks. 
This is due to a ``spam attack'' that interested the Zcash blockchain in 2022 and caused its size to increase greatly. 
This spam attack impacts also the database upgrade time, because generating history nodes requires reading back finalized blocks from the database, and while the mining rate of the blockchain remains approximately constant, each block is much larger due to the large number of transactions caused by the attack. 
This results in larger database read times over the span of the attack window. 
In general, we can see that outside the spam attack window both the upgrade time and full sync time increase linearly with the length of the blockchain, while the full sync time is impacted by the format change in a negligible way.

Figure \ref{fig:flyclient_pseudocode} shows our reference FlyClient verifier, described as pseudo-code that uses the aforementioned RPC functions. 
\begin{figure}[t]
\centering
\footnotesize

\subfloat[FlyClient verifier\label{alg:verifier}]{
\begin{minipage}{\linewidth}
\input{algorithms/alg_verifier}
\end{minipage}
}

\vspace{0.8em}

\subfloat[Checking a prover\label{alg:check_prover}]{
\begin{minipage}{\linewidth}
\input{algorithms/alg_check_prover}
\end{minipage}
}

\vspace{0.8em}

\subfloat[Sampling at a height\label{alg:sample_at_height}]{
\begin{minipage}{\linewidth}
\input{algorithms/alg_sample_at_height}
\end{minipage}
}

\caption{FlyClient reference verifier.}
\label{fig:flyclient_pseudocode}
\end{figure}
The \textproc{FlyClientVerifier} function (Figure \ref{alg:verifier}) tests the trustworthiness of a set of verifiers $P$, starting from the one that declared the chain with the most work. 
The \textproc{checkProver} function (Figure \ref{alg:check_prover}) tests a single prover $p$, by first performing $n_\mathit{det}$ deterministic samplings (Lines 3--9) and then $n_\mathit{prob}$ probabilistic ones (Lines 10--17). 
In doing this, it maintains two caches of downloaded headers ($H$) and MMR nodes ($M$), to avoid downloading them multiple times.
The \textproc{sampleWork} function (called in Line 12, not shown as pseudo-code) produces a random number between the two extremes specified as arguments, using the probability distribution specified by B\"unz et al. \cite{bunz2020flyclient} (Theorem 2).
Finally, the \textproc{sampleAtHeight} function (Figure \ref{alg:sample_at_height}) samples a header at a given height $h$ from prover $p$, and it is used for the probabilistic sampling phase.
To take into account the MMR reboots at the network upgrades, it checks if the header to sample lies before the current consensus branch, and it samples also the initial block of the current consensus branch in case (Lines 3--8).
We take into consideration also variants to this verifier, but we will omit their pseudo-code for the sake of brevity. Namely, we consider the following alternative verifiers.
\begin{itemize}
\item \emph{Fixed-difficulty verifier.} A FlyClient verifier that assumes a blockchain with fixed difficulty, and verifies the chain length in terms of number of blocks rather than their total work. 
This verifier runs a simplified protocol, because it skips the checks on difficulty transitions and avoids to use the difficulty-related RPC functions \textproc{getTotalWork} and \textproc{getHeightWithTotalWork}. 
Such a verifier is considered by the original FlyClient paper \cite{bunz2020flyclient}, but it is only for theoretical interest since it is highly insecure with a real-life blockchain. 
Indeed, a malicious prover could build a very long valid blockchain easily, by simply mining all the blocks with a very low difficulty.
\item \emph{Cache-less verifier.} A FlyClient verifier that does not maintain caches for headers and nodes already downloaded. 
This verifier could be suitable for very constrained devices that lack memory capacity. 
Of course it may download multiple times the same header or the same node, so in general the proof size increases with respect to the reference FlyClient verifier.
\end{itemize}

Fig. \ref{fig:proof_size_repr} shows the size of the interactive proof downloaded by the reference verifier under three different representations, namely JSON, binary, and zipped binary, varying the length of the chain. 
\begin{figure}[t]
    \centering
    \includegraphics[width=0.9\linewidth]{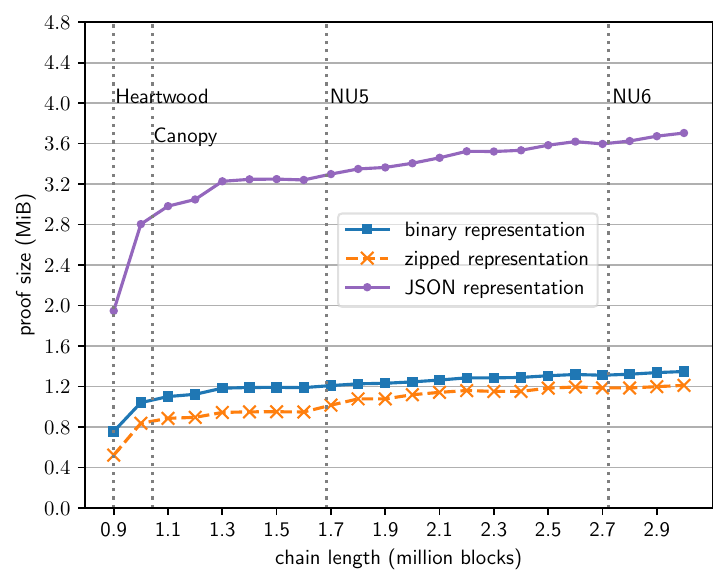}
    \caption{FlyClient proof size with respect to the proof representation.}
    \label{fig:proof_size_repr}
\end{figure}
In the JSON representation, both the MMR nodes and the block headers constituting the ancestry proofs are represented as textual JSON structures, while the Equihash solutions inside block headers are represented as long strings of hexadecimal digits. 
This is the basic and less efficient proof representation. 
In the binary representation, both the MMR nodes and the block headers are represented as bytes strings, using the native Zebrad representation. 
In the zipped representation, we applied a compression algorithm to the binary representation of each MMR node and each block header. 
In particular, we employed the gzip algorithm, as it is often done in the HTTP protocol for saving bandwidth. 
The proof size is computed as the total number of bytes downloaded by the verifier during the proving protocol. 
We measure only the payload size, that is we do not count the overhead of transport or application protocols like TCP or HTTP. 
The bytes uploaded by the verifier are negligible with respect to the downloaded ones, so we omit them in the present paper. 
We assumed an adversarial block budget of 50, corresponding to an expected budget of 20,833 USD (see Section \ref{sec:wa_adversary}).
To obtain statistical significance, every point of the figure is averaged over 30 independent repetitions of the same experiment. 
We computed the 95\%-confidence intervals but their width is in the order of 14--18 KiB, so we omitted them in the figures. 
The vertical lines are the heights of the various network upgrades. 
We can see that the binary representation more than halves the proof size compared to the JSON one, while the zipped representation saves few hundreds of KiB with respect to the binary one. 
The JSON representation is of course not recommendable for IoT applications that must save bandwidth. 
In every representation, the majority of such proof sizes is constituted by Equihash solutions, each of which takes 1344 bytes under the binary representation. 
All the proof sizes follow logarithmic trends, with the zipped representation being around 1.2 MiB at 3 million blocks, which were completed on July 20, 2025. 
Note that the various MMR rebootings due to the network upgrades have little impact on the proof sizes, as they perturb almost negligibly its logarithmic trend. 

Fig. \ref{fig:proof_size_ni_repr} shows the size of the non-interactive proof under the same three representations, varying the length of the chain. 
\begin{figure}[t]
    \centering
    \includegraphics[width=0.9\linewidth]{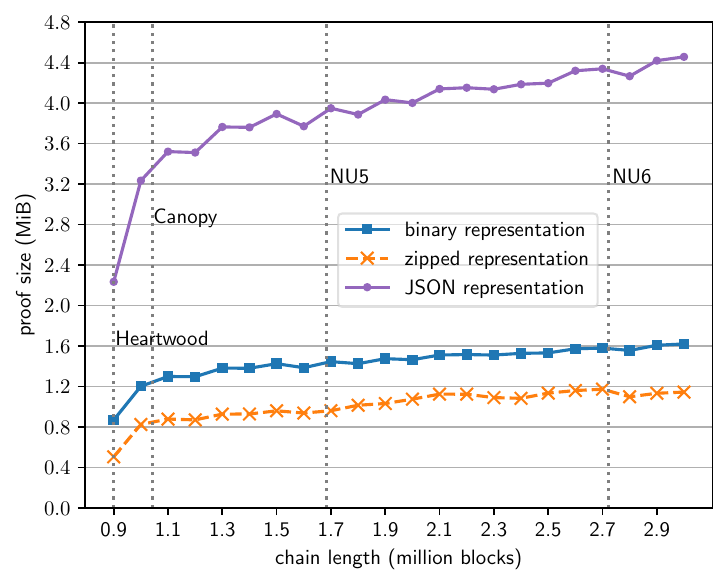}
    \caption{Non-interactive FlyClient proof size with respect to the proof representation.}
    \label{fig:proof_size_ni_repr}
\end{figure}
This time, in the zipped representation we applied the gzip compression algorithm over the \emph{whole} proof, rather than the single MMR nodes and block headers. 
This is to emulate a realistic use case in which an interactive proof is downloaded by separate HTTP requests, so it is compressed piece-wise, while a non-interactive proof is sent as a unique compressed file. 
All the proof sizes confirm the logarithmic trends and the little impacts of network upgrades. The JSON and the binary representations are some hundreds of KiB larger than their interactive counterparts, which constitutes the cost of the non-interactiveness. 
On the other hand, the zipped representation reaches quite the same size of the interactive counterpart, that is 1.2 MiB at 3 million blocks. 
This probably stems from the better effectiveness of the compression when applied to the proof as a whole. 
Assuming an hypothetical trustless bridge application, a FlyClient proof should be passed to a smart contract, possibly fragmented in several method calls. 
Therefore, the proof size translates directly to a monetary cost. 
After the Pectra hard fork of May 2025, a data-heavy transaction is charged using the EIP-7623 \cite{eip7623}, which specifies 40 gas per non-zero byte of calldata. 
Therefore, passing a FlyClient proof of 1.2 MiB to a smart contract takes approximately 50,340,000 gas. 
Considering a gas price of 0.125 Gwei per gas\footnote{\url{https://etherscan.io/gastracker} (accessed on 7 April 2026). These values are taken as a snapshot and may vary over time.} and an ETH price of 2,100 USD per ETH\footnote{\url{https://etherscan.io/gastracker} (accessed on 7 April 2026).}, this takes 50,340,000 $\times$ 0.125 $\times$ $10^{-9}$ $\times$ 2,100 = 13.21 USD. 
In other words, we are using only 13.21 USD to defend against an adversary that employs 20,833 USD.
This should be acceptable for infrequent high-value cross-chain transfers, but less acceptable for frequent or low-value ones. 
In Section \ref{sec:optimizations} we will see how to decrease the proof size and thus the gas cost of a trustless bridge. 
Note that we do not consider to store the FlyClient proof as EIP-4844 blobs \cite{eip4844}. 
This is on purpose, because even if blobs would dramatically decrease the costs, 
they would also weaken the trust model. 
Indeed, a smart contract cannot directly access blob data, therefore validators would not check the proof to be correct.
From now on, we will consider the zipped representation only, as it is the more efficient one and it has no side costs on other aspects.

Fig. \ref{fig:proof_size_verifier} shows the size of the proof with the three verifiers, represented in zipped binary format, varying the length of the chain.
\begin{figure}[t]
    \centering
    \includegraphics[width=0.9\linewidth]{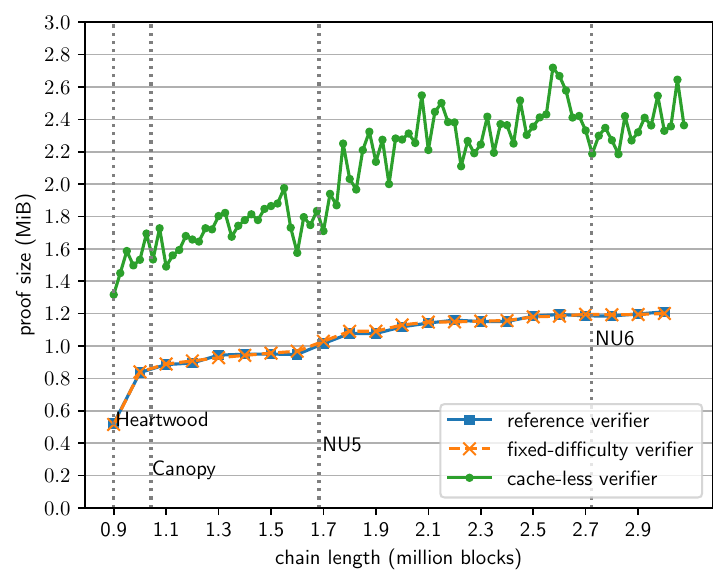}
    \caption{FlyClient proof size with respect to the type of verifier.}
    \label{fig:proof_size_verifier}
\end{figure}
We can see that the proof size of the fixed-difficulty verifier closely follows the one of the reference verifier. 
This can be explained by the fact that PoW difficulty has increased quite slowly in Zcash, so sampling blocks based on their height turns out to be quite similar to doing it based on their cumulative work. 
On the other hand, the proofs of the non-cached verifier are sensibly bigger than the ones of the reference verifier. 
The proof size of this verifier is also non-monotone with the length of the chain, as its plot presents peaks and valleys. 
This can be explained by the fact that the non-cached verifier must download multiple times the peak nodes of the ``old'' mountains of the MMR, whose number changes in a non-monotone way with the chain length. 
For example, an MMR covering 255 blocks has 8 peak nodes, while an MMR covering just one more block has only 1 peak node. 
In the reference verifier this phenomenon also happens, but its effect is mitigated by the node caching, that allows the verifier to download the peak nodes of the old mountains only once.

Fig. \ref{fig:proof_size_ni_verifier} shows the size of the non-interactive proof with the reference and the fixed-difficulty verifiers, varying the length of the chain.
We omitted the cache-less verifier here because it does not apply to non-interactive proofs.
\begin{figure}[t]
    \centering
    \includegraphics[width=0.9\linewidth]{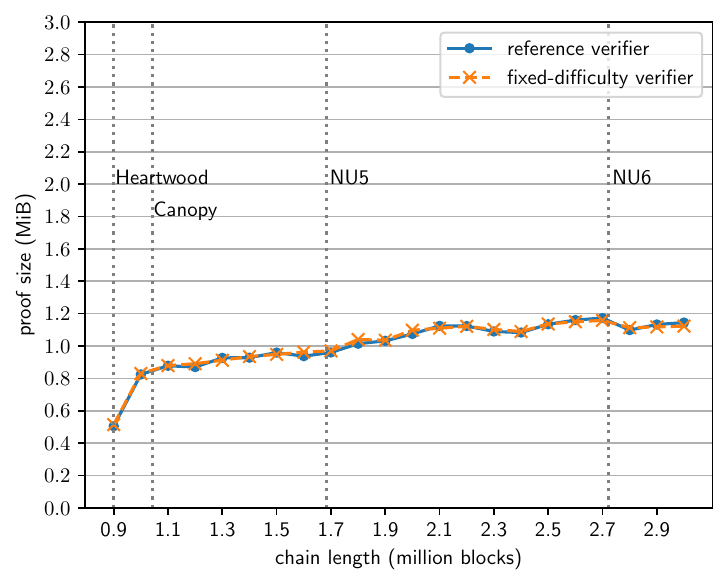}
    \caption{Non-interactive FlyClient proof size with respect to the type of verifier.}
    \label{fig:proof_size_ni_verifier}
\end{figure}
All the general trends of the interactive proofs (see Fig. \ref{fig:proof_size_verifier}) are confirmed. 
Even here, the proof size of the fixed-difficulty verifier closely follows the one of the reference verifier.

%% file: algorithms/alg_verifier.tex
\begin{algorithmic}[1]
\Function{FlyClientVerifier}{$P, n_{\mathit{det}}, n_{\mathit{prob}}$}
  \ForAll{$p \in P$}
    \State $(n_p,w_p,t_p)\gets p.\textsc{getBlockchainInfo}()$
  \EndFor
  \ForAll{$p \in P$ in decreasing order of $w_p$}
    \If{\Call{checkProver}{$p,n_p,t_p,n_{\mathit{det}},n_{\mathit{prob}}$}}
      \State \Return $(p,n_p,w_p,t_p)$
    \EndIf
  \EndFor
  \State \Return $\bot$
\EndFunction
\end{algorithmic}

%% file: algorithms/alg_check_prover.tex
\begin{algorithmic}[1]
\Function{checkProver}{$p,n,t,n_\mathit{det},n_\mathit{prob}$}
  \State $H\gets\emptyset,\; M\gets\emptyset$
  \For{$h = n - n_{\mathit{det}} - 1,\dots,n$}
    \State $B\gets p.\textsc{getBlockHeader}(h)$
    \If{invalid difficulty transition in $B$ or invalid PoW in $B$ or $B$ not chained to $t$}
      \State \Return \textbf{false}
    \EndIf
    \State $H[h]\gets B$
  \EndFor
  \State $w_\mathit{det} \gets p.\textsc{getTotalWork}(n-n_{\mathit{det}}-1)$
  \For{$j=1,\dots,n_{\mathit{prob}}$}
    \State $x\gets \textsc{sampleWork}(0,w_\mathit{det})$
    \State $h\gets p.\textsc{getHeightWithTotalWork}(x)$
    \If{\textbf{not} \Call{sampleAtHeight}{$p,h,t,H,M$}}
      \State \Return \textbf{false}
    \EndIf
  \EndFor
  \State \Return \textbf{true}
\EndFunction
\end{algorithmic}

%% file: algorithms/alg_sample_at_height.tex
\begin{algorithmic}[1]
\Function{sampleAtHeight}{$p,h,t,H,M$}
  \State $u\gets$ last network-upgrade height before $t$
  \If{$h<u$}
    \If{\textbf{not} \Call{sampleAtHeight}{$p,h,u,H,M$}}
      \State \Return \textbf{false}
    \EndIf
    \State $h\gets u$
  \EndIf
  \LineIf{$h\in H$}{\Return \textbf{true}}
  \State $B\gets p.\textsc{getBlockHeader}(h)$
  \LineIf{invalid PoW in $B$}{\Return \textbf{false}}
  \State $I\gets$ node indices of ancestry proof for $h$ in the MMR rooted at $t$
  \ForAll{$i\in I\setminus \mathrm{dom}(M)$}
    \State $M[i]\gets p.\textsc{getHistoryNode}(i)$
  \EndFor
  \State $\pi\gets M|_I$
  \If{implausible difficulty changes in $\pi$ or $\pi$ invalid proof for $(B,t)$}
    \State \Return \textbf{false}
  \EndIf
  \State \Return \textbf{true}
\EndFunction
\end{algorithmic}

%% file: sections/06_optimizations.tex
\section{Optimizations}
\label{sec:optimizations}

We introduce here two optimizations of the reference verifier of Fig. \ref{fig:flyclient_pseudocode} that aim at minimizing the size of the interactive and non-interactive proofs. 
The objective is to improve the bandwidth efficiency of IoT devices that execute the interactive FlyClient protocol, but also to minimize the cost of storing a non-interactive FlyClient proof on a smart contracts for trustless bridge applications. 
The two optimizations are as follows: (1) \emph{cumulative proof}, which consists in downloading a single MMR proof that proves the inclusion in the blockchain of all the sampled blocks altogether, and (2) \emph{distilled proof}, which requires a change in the consensus algorithm specifically aimed at saving FlyClient proof size. 
We describe them both in the following sections.

\subsection{Cumulative Proof}
The FlyClient proof as currently described by the literature \cite{bunz2020flyclient,nemoz2021deployment,perazzo2024smartfly} and downloaded by the reference verifier of Fig. \ref{fig:flyclient_pseudocode} is suboptimal from the size point of view. 
The reason is that, in many cases, the verifier downloads nodes that were computable offline from other already downloaded nodes or headers. 
As an example, let us consider the toy MMR of Fig. \ref{fig:cumulative_proof}, and let us assume that the sampled headers are $B_3$ and $B_6$. 
\begin{figure}[t]
    \centering
    \includegraphics[trim={2.65cm 8.06cm 21.27cm 2.61cm},clip,width=0.9\linewidth]{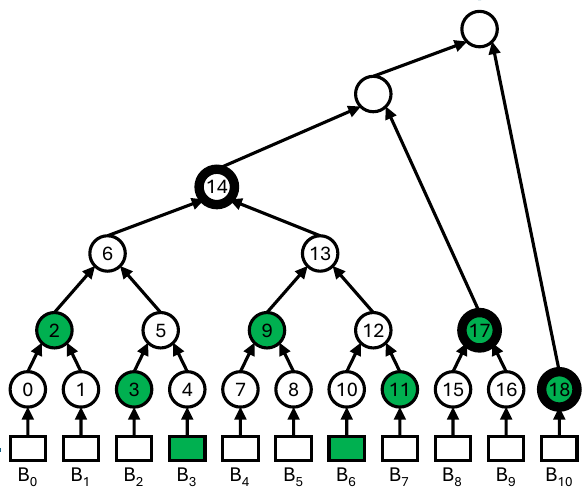}
    \caption{Example of cumulative proof of ancestry.}
    \label{fig:cumulative_proof}
\end{figure}
The reference verifier would download separately the ancestry proof relative to $B_3$ and that relative to $B_6$. 
Both proofs consist of four nodes, respectively 2, 3, 13, 17, 18 for the $B_3$'s proof and 6, 9, 11, 17, 18 for the $B_6$'s one. 
Considering that nodes 17 and 18 are cached after the download of the first proof, the total number of nodes downloaded by the reference verifier is 8. 
However, some of these nodes were computable from others. 
For example, node 6 was computable from nodes 2, 3, and the hash of header $B_3$.
Similarly, node 13 was computable from nodes 9, 11, and the hash of header $B_6$. To capture this optimization, we introduce the concept of cumulative ancestry proof.

Given a set of sampled headers, their cumulative ancestry proof is constituted by the minimal set of nodes that covers only the spaces between the sampled headers. 
In the example of Fig. \ref{fig:cumulative_proof}, the cumulative ancestry proof consists of nodes 2, 3, 9, 11, 17, 18 (the green MMR nodes in the figure). 
A verifier that downloads the cumulative ancestry proof will thus download 6 nodes instead of 8. 
It is easy to show that we can reconstruct the chain history roots relative to the top header and to all the sampled headers from the sampled headers themselves plus their cumulative ancestry proof. 
In addition, it is possible to check the plausibility of the difficulty variations from the cumulative ancestry proof, since it covers the whole blockchain. 
Thus, the cumulative ancestry proof is an efficient replacement for the singular ancestry proofs of the sampled blocks, and it has no side costs on other aspects. 
Note that this optimization decreases the number of nodes to be downloaded by the verifier, but it does not decrease the size of both nodes and headers. 
To do this, we must change their format, which in turn implies a change in the consensus algorithm. 
This will be the aim of the second optimization.

\subsection{Distilled Proof}
The current format of the Zcash headers is particularly inefficient for deploying FlyClient verifiers. 
This is mainly because each header contains an Equihash solution, which occupies 1344 bytes. 
Hence, the verifier must download at least $(n_\mathit{det}+n_\mathit{prob}) \cdot 1344$ bytes. 
The situation would improve if Zcash used a different PoW function, for example a double-SHA256 like Bitcoin, which does not require to store any solution inside headers. 
The reason why Zcash uses Equihash is to obtain a certain degree of ASIC-resistance, due to the memory-hardness of Equihash. 
However, other memory-hard hash functions have been designed, which do not require solution storage, for example Argon2 and its variants \cite{biryukov2016argon2}. 
The Zcash community is evaluating to replace Equihash with other more space-efficient hashes\footnote{\url{https://gist.github.com/zmanian/adb4d41a725fe64ccfeb9a435e6d22bd} (accessed on 25-03-2026).}, but no consensus has been reached yet.

A distilled proof is a cumulative proof that also minimizes the quantity of data that a FlyClient verifier must download to perform all the necessary checks. 
This requires applying three modifications to the current Zcash consensus algorithm, so it is not immediately deployable.
First, the memory-hard problem must be changed in such a way to admit small solutions. 
As we said before, an Equihash solution is 1344 bytes in size, and it must be included in each header downloaded by the verifier to allow the verifier to check its correctness. 
The most radical choice would be to get rid of the memory-hard problem so that 
no solution needs to be stored in the headers.
The Bitcoin consensus works in this way. 
Of course, this choice would jeopardize completely the ASIC-resistance of the consensus. 
If we want to save ASIC-resistance, a trade-off is to adopt a memory-hard problem that requires small solutions, like Ethash \cite{wood2014ethereum} adopted by the pre-Merge Ethereum and by Ethereum Classic\footnote{\url{https://ethereumclassic.org/} (accessed on 25-03-2026).}, or RandomX\footnote{\url{https://github.com/tevador/RandomX/blob/master/doc/specs.md} (accessed on 25-03-2026).} adopted by Monero. 
We assume to adopt Ethash as a PoW algorithm, whose memory-hard problem solution (called \emph{mixhash}) is 32 bytes only. 
In this way, a Zcash header will decrease from 1487 bytes to 175.

The second modification requires that the final hash whose result is compared to the target is computed not on the header, but on its hash plus those fields that are strictly necessary to the FlyClient protocol, that are the memory-hard problem solution, the chain history root, the target, and the timestamp. 
This allows the verifier to avoid downloading the actual header, but only a ``distilled'' version of it, which includes all the fields necessary to perform the FlyClient checks. 
Assuming to use SHA-256 to compute the header's hash, each distilled Zcash header will decrease again from 175 bytes to 104. Note that this trick is applicable to the generic PoW-consensus cryptocurrency, making it ``FlyClient-friendly'' with virtually no cost. 
Note also that the header distillation can boost the performance of SPV light clients too, since they can download the distilled headers instead of the full ones. 
Of course, in this case each distilled header must also contain an extra 32-byte field for the previous block's hash.

The third modification requires that the MMR nodes are ``distilled'' too, in the sense that they contain the fields strictly necessary to the FlyClient protocol, and only a hash of the other fields. 
Indeed, the fields necessary for FlyClient are child nodes' hash, the start and the end heights, the start and the end difficulties, the start and the end timestamps, and the cumulative work. 
Zcash MMR nodes contain some additional fields which are not strictly necessary for the FlyClient checks, which are the start and the end Sapling roots, the number of Sapling transactions, and from NU5 onward the start and the end Orchard roots, and the number of Orchard transactions. 
Hashing the unnecessary fields allow the verifier to avoid downloading them during the FlyClient protocol, and to download them afterwards if interested. 
Assuming to use SHA-256 to compute the fields' hash, each distilled MMR node will decrease from 244 bytes to 140. 

\subsection{Experimental Evaluation}
In this section we experimentally evaluate the performance improvements carried by the previously described optimizations to the size of the FlyClient proof. 
Fig. \ref{fig:proof_size_opt} shows the interactive proof size of the reference verifier represented in zipped binary format respectively without any optimization, with cumulative proof, and with cumulative distilled proof, varying the length of the chain. 
\begin{figure}[t]
    \centering
    \includegraphics[width=0.9\linewidth]{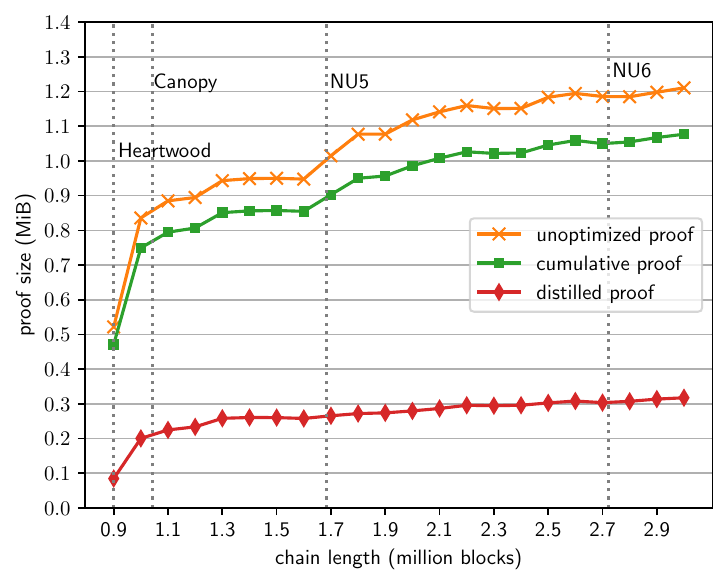}
    \caption{FlyClient proof size under different optimizations.}
    \label{fig:proof_size_opt}
\end{figure}

For the distilled proof, we simulated that the modifications to the PoW consensus algorithm were applied just after the Heartwood network upgrade, so that we can conveniently compare the distilled proof with the other proofs. 
In a real deployment, consensus modifications would be applied during a future network upgrade, but the performance improvement will be eventually the same. 
It can be seen that the cumulative proof carries only a measurable but moderate reduction of the proof size, while the distilled one carries a bigger improvement. 
Of course, the downside of the distilled proof is that it requires a hard fork, while the cumulative proof is applicable just now without changing anything in the consensus algorithm.
Fig. \ref{fig:proof_size_ni_opt} shows the non-interactive proof size, varying the length of the chain.  
\begin{figure}[t]
    \centering
    \includegraphics[width=0.9\linewidth]{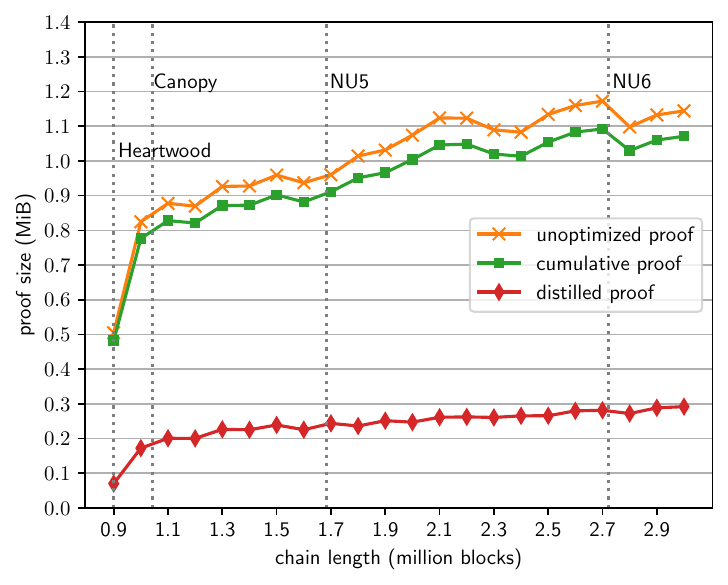}
    \caption{Non-interactive FlyClient proof size under different optimizations.}
    \label{fig:proof_size_ni_opt}
\end{figure}
All the trends of the interactive proof are roughly confirmed.
Note that the increase in proof size due to the non-interactiveness is well compensated by the whole-proof compression.
The cumulative proof saves less space than in the interactive case. This can be explained because the whole-proof compression is already very efficient in saving space.
A non-interactive cumulative proof is below 1.09 MiB with 3 million blocks, which takes approximately 
45,720,000 gas to pass it to a smart contract for trustless bridge applications, equivalent to 12.00 USD (-9.16\% compared to not using cumulative proofs).
The non-interactive distilled proofs are even smaller than the interactive ones.
As an example, a non-interactive distilled proof of a chain with 3 million blocks is 320 KiB, which takes approximately 
13,110,000 gas to pass it to a smart contract, equivalent to 3.44 USD (-71.33\% compared to cumulative proofs without distillation).

%% file: sections/07_conclusions.tex
\section{Conclusions and Future Work}
\label{sec:conclusions}

This paper takes a concrete step toward bridging the gap between the theoretical design of FlyClient and its practical deployment in real-world blockchain systems. 
We addressed both modeling and implementation challenges, providing insights that are directly relevant for production environments.
First, we introduced the $w_a$-adversary model as an alternative to the classical $(c,L)$-adversary. 
By expressing the adversary in terms of an economic budget, the model offers a more intuitive and deployment-oriented perspective, enabling practitioners to parametrize FlyClient verifiers through a single, interpretable parameter. 
We also showed how such a model can be used to derive verifier configurations via an optimization procedure, which saves some proof size as side effect.
Second, we designed and implemented a FlyClient prover for the Zcash blockchain by extending the Zebrad full node. 
Our implementation demonstrates that FlyClient functionalities can be integrated into an existing production-grade system with negligible impact on synchronization time, while requiring only moderate additional storage. 
This provides, to the best of our knowledge, the first end-to-end evaluation of FlyClient in a real PoW blockchain.
We also analyzed the size of interactive and non-interactive FlyClient proofs under different representations and verifier variants, highlighting the dominant impact of block header components (notably Equihash solutions) on proof size. 
Third, we proposed two optimizations: cumulative proofs, which reduce redundancy in ancestry proofs without protocol changes, and distilled proofs, which significantly decrease proof size at the cost of some modifications to the consensus layer. 
Our experimental evaluation shows that distilled proofs can reduce non-interactive proof sizes to a few hundred KiB, making on-chain verification even more economically viable.
Overall, our results indicate that FlyClient is a promising building block for applications such as trustless cross-chain bridges and lightweight verification in constrained environments. 
In future work, we plan to study the application of SNARKs to non-interactive FlyClient proofs, to make them even more compact and efficient from the verifier point of view.